\documentclass[orivec]{llncs}

\usepackage{cite}
\usepackage{xspace}
\usepackage{amssymb}
\usepackage{amsmath}
\usepackage{amsthm}
\usepackage{wasysym}
\usepackage{multirow}
\usepackage{thm-restate}

\usepackage{tikz}
\usepackage{pgfplots} 
\pgfplotsset{compat=newest}
\usetikzlibrary{calc}

\usepackage{hyperref}

\iftrue
	\usepackage{todonotes}
	
	\newcommand{\at}[1]{\todo[inline,color=teal!10,caption={AT}]{\textbf{AT:} #1}}
	\newcommand{\emh}[1]{\todo[inline,color=orange!10,caption={EMH}]{\textbf{EMH:} #1}}
	\newcommand{\ly}[1]{\todo[inline,color=orange!10,caption={LY}]{\textbf{LY:} #1}}
	\newcommand{\lz}[1]{\todo[inline,color=red!10,caption={ZLj}]{\textbf{ZLj:} #1}}
	\newcommand{\xb}[1]{\todo[inline,color=orange!10,caption={XB}]{\textbf{XB:} #1}}
\else
	\newcommand{\at}[1]{}
	\newcommand{\emh}[1]{}
	\newcommand{\ly}[1]{}
	\newcommand{\lz}[1]{}
	\newcommand{\xb}[1]{}
\fi

\renewcommand{\epsilon}{\varepsilon}
\renewcommand{\phi}{\varphi}

\renewcommand{\vec}[1]{\mathbf{#1}}

\newcommand{\bigO}{\mathcal{O}}
\newcommand{\naturals}{\mathbb{N}}

\newcommand{\reals}{\mathbb{R}}
\newcommand{\posreals}{\reals_{\geq 0}}

\newcommand{\abs}[1]{\lvert #1 \rvert}
\newcommand{\norm}[2][]{\lVert #2 \rVert_{#1}}

\newcommand{\upperboundFXBeta}{\mathit{UB}}

\newcommand{\interval}[2]{[#1,#2]}

\newcommand{\size}[1]{|#1|}

\newcommand{\setnocond}[1]{\{#1\}}
\newcommand{\setcond}[2]{\{\, #1 \,|\, #2 \,\}}

\newcommand{\dirac}[1]{\delta_{#1}}

\newcommand{\mystackrelsingle}[2]{%
  \mathrel{\vbox{\offinterlineskip\ialign{%
    \hfil##\hfil\cr%
    $\scriptstyle#1$\cr%
    \noalign{\kern.3ex}%
    $#2$\cr%
}}}}

\newcommand{\AP}{\mathit{AP}}

\newcommand{\modelsymbol}[1]{\mathcal{#1}}
\newcommand{\dtmc}[1][D]{\modelsymbol{#1}}
\newcommand{\pdtmc}[1][D]{\dtmc[#1]_{\parametersSet}}
\newcommand{\dtmrm}[1][R]{\modelsymbol{#1}}
\newcommand{\pdtmrm}[1][R]{\dtmc[#1]_{\parametersSet}}

\newcommand{\mstates}{S}
\newcommand{\minit}[1][s]{\bar{#1}}

\newcommand{\mtransitions}{\mathbf{P}}
\newcommand{\mlabelling}{L}
\newcommand{\mreward}[1][r]{\mathfrak{#1}}

\newcommand{\mpath}{\pi}
\newcommand{\mpaths}[1]{\mathit{Paths}(#1)}
\newcommand{\mpathsfin}[1]{\mathit{Paths}^{*}(#1)}

\newcommand{\cylinder}[1]{\mathit{Cyl}(#1)}
\newcommand{\prob}{\mathit{Pr}}
\newcommand{\exprew}{\mathit{ExpRew}}

\newcommand{\parameter}[1][v]{#1}
\newcommand{\parameters}[1][\parameter]{\vec{#1}}
\newcommand{\parametersSet}{\mathrm{V}}

\newcommand{\ratfunsSet}[1][\parametersSet]{\mathcal{F}_{#1}}
\newcommand{\polynomialsSet}[1][\parametersSet]{\mathcal{P}_{#1}}
\DeclareMathOperator{\range}{range}
\newcommand{\evaluation}{\nu}
\newcommand{\evaluate}[2][\evaluation]{#2\langle#1\rangle}

\newcommand{\success}{\checked}
\newcommand{\failure}{\mathord{\times}}

\newcommand{\lsf}{\phi}
\newcommand{\lpf}{\psi}
\newcommand{\lP}[1]{\mathtt{P}_{#1}}
\newcommand{\lER}[1]{\mathtt{R}_{#1}}
\newcommand{\lX}{\mathbf{X}}
\newcommand{\lU}{\mathbin{\mathbf{U}}}
\newcommand{\lBU}[1]{\mathbin{\mathbf{U}^{\leq #1}}}
\newcommand{\ltrue}{\mathtt{tt}}

\newcommand{\lF}{\mathbf{F}}

\newcommand{\errorRate}{\epsilon}
\newcommand{\significanceLevel}{\eta}
\newcommand{\margin}{\lambda}

\newcommand{\safetyLevel}{\zeta}
\newcommand{\rewardLevel}{\rho}

\newcommand{\approxfun}[1]{\tilde{#1}}

\newcommand{\pacmc}[1][]{\textsc{Tool}\ensuremath{^{#1}}\xspace}
\newcommand{\storm}{\textsc{Storm}\xspace}
\newcommand{\prism}{\textsc{PRISM}\xspace}

\newcommand{\matlab}{\textsc{MATLAB}\xspace}

\newcommand{\benchexec}{\textsc{BenchExec}\xspace}

\spnewtheorem{assumption}{Assumption}{\bfseries}{\itshape}

\title{Scenario Approach for Parametric Markov Models}

 \author{
 Ying Liu\inst{1,2}
 \and
 Andrea Turrini\inst{1,3}
  \orcidID{0000-0003-4343-9323}
 \and
 Moritz Hahn\inst{4}
 \and
 Bai Xue\inst{1}
 \and
 Lijun Zhang\inst{1,2,3}
  \orcidID{0000-0002-3692-2088}
 }

 \institute{
 State Key Laboratory of Computer Science, 
 Institute of Software, Chinese Academy of Sciences, China
 \and
 University of Chinese Academy of Sciences, China
 \and
 Institute of Intelligent Software Guangzhou, China
 \and Formal Methods and Tools, University of Twente, Enschede, The Netherlands
 }

\allowdisplaybreaks
\begin{document}

\maketitle

\setcounter{footnote}{0}

\begin{abstract}
    In this paper, we propose an approximating framework for analyzing parametric Markov models. Instead of computing complex rational functions encoding the reachability probability and the reward values of the parametric model, we exploit the scenario approach to synthesize a relatively simple polynomial approximation. The approximation is probably approximately correct (PAC), meaning that with high confidence, the approximating function is close to the actual function with an allowable error. With the PAC approximations, one can check properties of the parametric Markov models. We show that the scenario approach can also be used to check PRCTL properties directly -- without synthesizing the polynomial at first hand.  We have implemented our algorithm in a prototype tool and conducted thorough experiments. The experimental results demonstrate that our tool is able to compute polynomials for more benchmarks than state-of-the-art tools such as \prism and \storm, confirming the efficacy of our PAC-based synthesis.
\end{abstract}

\section{Introduction}
\label{sec:introduction}

Markov models (see, e.g.,~\cite{DBLP:books/wi/Puterman94}) have been widely applied to reason about quantitative properties in numerous domains, such as networked, distributed systems, biological systems~\cite{von2006five}, and reinforcement learning~\cite{DBLP:journals/ml/WatkinsD92, DBLP:conf/icra/BaiCYHL15}. 
Properties analyzed on Markov models can either be simple, such as determining the value of the probability that a certain set of unsafe states is reached and how an expected reward value compares with a specified threshold, or complex, involving employing temporal logics such as PCTL~\cite{DBLP:journals/fac/HanssonJ94,DBLP:conf/fsttcs/BiancoA95} and PRCTL~\cite{DBLP:conf/formats/AndovaHK03}.
To verify these properties, various advanced tools have been developed, such as
\prism~\cite{DBLP:conf/cav/KwiatkowskaNP11}, \storm~\cite{DBLP:journas/sttt/HenselJKQV21,DBLP:conf/cav/DehnertJK017}, MRMC~\cite{DBLP:journals/pe/KatoenZHHJ11}, CADP 2011~\cite{DBLP:journals/sttt/GaravelLMS13}, PROPhESY~\cite{DBLP:conf/cav/DehnertJJCVBKA15} and  \textsc{IscasMc}~\cite{DBLP:conf/fm/HahnLSTZ14}. 

In this paper we consider \emph{parametric} discrete time Markov chains (pDTMCs), whose transition probabilities are not required to be constants, but can depend on a set of parameters. 
For this type of models, the value of the analyzed property can be described as a \emph{function} of the parameters, mapping either to truth values or to numbers.
In many cases, these functions are \emph{rational functions}, that is, fractions of co-prime polynomials.
The exact rational function is commonly challenging to compute as it often involves polynomials with very high degree~\cite{DBLP:journals/iandc/BaierHHJKK20}.

\medskip\noindent\textbf{Contribution of the paper.}
In this work, we propose an alternative approach to obtain the function $f_{\lsf}$ describing the value of the analyzed property $\lsf$ in the given pDTMC. 
The main idea is to learn a polynomial with low degree to approximate the actual function $f_{\lsf}$ in pDTMC and pDTMRM.
Exploiting the scenario approach~\cite{DBLP:journals/tac/CalafioreC06,DBLP:journals/arc/CampiGP09}, we obtain an \emph{approximating function} with probably approximately correct (PAC) guarantee, i.e., with high confidence $1 - \significanceLevel$, the probability that the approximation is within an error margin $\margin$ is at least $1 - \errorRate$.
The obtained function can then be used in synthesizing parameters and analyzing properties of $f_{\lsf}$. 

Experimental results we performed show that our prototype \pacmc can solve more properties under the same conditions than the state-of-the-art verification tools \storm and \prism, and provide PAC approximations with statistical guarantee. 
We can use the PAC approximation to check the safe region of the parameter space and other properties, such as whether the probabilities of satisfying the given PRCTL formulas remain within a certain range on changing parameters' values.
Also, compared with the Taylor expansion to approximate the actual function $f_{\lsf}$, PAC approximation can approximate $f_{\lsf}$ more accurately: 
the $L_{2}$ norm of the distance between $f_{\lsf}$ and the PAC approximation can be even hundreds of times smaller than that of the Taylor expansion.
We also extend the scenario approach to reward properties; 
we use PAC approximation to estimate the lower bound of the expectation of $f_{\lsf}$ with respect to a probability measure $P$ over the domain $X$ of the parameters: the quadratic polynomial can estimate the expectation $\int_{X} f_{\lsf}(\parameters)\,dP(\parameters)$ accurately while being easy to compute.

\medskip\noindent\textbf{Related work.}
Model checking of parametric Markov models is not a new area and a number of related works exist, each with different strengths and weaknesses.
In the following, we demarcate our work from the existing ones.

Daws has devised a language-theoretic approach to solve the reachability problem in parametric Markov chains~\cite{DBLP:conf/ictac/Daws04}.
In this approach, the model is viewed as a finite automaton. 
Based on the state elimination approach~\cite{DBLP:books/daglib/0016921}, the regular expression describing the language of such an automaton is computed.
In a postprocessing step, this regular expression is transformed into a rational function over the parameters of the model. 

In a following work~\cite{DBLP:journals/sttt/HahnHZ11}, the method has been improved by intertwining the state elimination and the computation of the rational function.
This improved algorithm has been implemented in the tool PARAM~\cite{DBLP:conf/cav/HahnHWZ10}.
PARAM also supports bounded reachability, relying on matrix-vector multiplication with rational function entries, and reachability rewards~\cite{key0126090m,DubinsS65}.
For the latter, the model is extended with parametric rewards assigned to both states and transitions.
Thereby, one can consider the expected accumulated reward until a given set of states is reached.
All these works~\cite{DBLP:journals/sttt/HahnHZ11,DBLP:conf/cav/HahnHWZ10} compute the precise rational function that describes the property of interest.
Unfortunately, it is challenging to evaluate it, due to the large coefficients and high exponents.
Moreover, the works discussed above do not consider properties specified by a temporal logic.

Several improvements have been proposed in later works.
Jansen et~al.~\cite{DBLP:conf/qest/JansenCVWAKB14}
perform the state elimination in a more systematic order, often leading to better performance in practice.
The work~\cite{DBLP:conf/atva/GainerHS18} uses arithmetic circuits, which are DAG-like structures, to represent such rational functions. 
A further work~\cite{DBLP:conf/nfm/HahnHZ11} follows a related approach to solve (potentially nested) PRCTL formulas for Markov decision processes:
the state-space is divided into hyperrectangles, and one has to show that a particular decision is optimal for a whole region. 
The work~\cite{DBLP:journals/iandc/BaierHHJKK20} improves the computation of the rational function by means of a fraction-free Gaussian elimination; 
the experimental evaluation confirms its effectiveness.
There are also methods for checking parametric continuous time Markov chains~\cite{DBLP:conf/gi/Han09}, by using a scenario approach~\cite{DBLP:conf/cav/BadingsJJSV22} or by being based on Gaussian processes~\cite{DBLP:conf/tacas/BortolussiS18,DBLP:journals/iandc/BortolussiMS16}.

The scenario approach was first introduced in~\cite{DBLP:journals/mp/CalafioreC05}, based on constraint sampling to deal with uncertainty in optimization. 
The works~\cite{DBLP:journals/tac/CalafioreC06,DBLP:journals/arc/CampiGP09,DBLP:journals/siamjo/CareGC15} study a probabilistic solution framework for robust properties. 
The work~\cite{DBLP:journals/siamjo/CareGC15} considers the min-max sample-based uncertain
convex optimization problems in the presence of stochastic uncertainty, which is called the “min-max scenario program”.
The work~\cite{DBLP:journals/tac/MargellosGL14} proposes a method to solve chance constrained optimization problems lying between robust optimization and scenario approach, which does not require prior knowledge of the probability distribution of the parameters.
The work~\cite{DBLP:journals/jota/CampiG11} based on~\cite{DBLP:journals/mp/CalafioreC05, DBLP:journals/tac/CalafioreC06} allows violating some of the sampled constraints
in order to improve the optimization value, and the work~\cite{DBLP:journals/automatica/VayanosKR12} expands the scenario optimization problem to multi-stage problems. 
Recently, the scenario approach has been applied to verify safety properties of black-box continuous time dynamical systems~\cite{xue2020pac} and the robustness of neural networks~\cite{DBLP:conf/icse/LiYHS0Z22}.

The most related to our work is~\cite{DBLP:journals/sttt/BadingsCJJKT22}, which also applies the scenario approach for analyzing parametric Markov chains and Markov decision processes. 
The main difference with our work is that in~\cite{DBLP:journals/sttt/BadingsCJJKT22}, the authors compute the probability that the instances of the parametric MDP satisfy a given property $\lsf$ with PAC-guarantee, by sampling the parameter values according to some unknown distribution; 
each MDP instance is then checked independently with respect to $\lsf$.
Instead, our work targets at computing an approximation of the complicated  function $f_{\lsf}$ --such as the one corresponding to the reachability probability $\lsf$-- depending on the parameters; 
we obtain this by sampling instances of the parameter values to compute the value of $f_{\lsf}$ on them and then synthesize the approximation with a certain confidence. 
Our framework can  bound the error between the actual function and the approximation we compute.
Moreover, as a side result, our PAC approximations can be used for visualizing the reachability probability, finding counterexamples, and analyzing properties that the original functions may satisfy.
Extending our approach to parametric MDPs seems feasible, as long as we treat the MDP strategy as in~\cite{DBLP:journals/sttt/BadingsCJJKT22}, i.e., we allow the strategy to change for the different MDP instances; 
that is, the strategy can also depend on the parametric values while solving the instantiated MDP with respect to $\lsf$.
We leave the formalization of the extension to parametric MDPs to future work.

\medskip\noindent\textbf{Organization of the paper.}
After giving in Sect.~\ref{sec:preliminaries} some preliminaries, models, and logic we use in this paper, in Sect.~\ref{sec:synthesisPACfunctions} we present our PAC-based model checking approach; 
we evaluate it empirically in Sect.~\ref{sec:experiments} before concluding the paper in Sect.~\ref{sec:conclusion} with some final remarks.

Due to space constraints, non-trivial proofs are provided in the appendix.

\section{Preliminaries}
\label{sec:preliminaries}

In this section, we first recall DTMCs, a well-know probabilistic model (see, e.g.,~\cite{DBLP:books/BaierK08PrinciplesMC}), reward structures, the probabilistic logic PRCTL we adopt to express properties on them, and then consider their extension with parameters.

\subsection{Probabilistic Models}
\label{ssec:probabilisticModels}

\begin{definition}
\label{def:dtmc}
    Given a finite set of atomic propositions $\AP$, a \emph{(labelled) discrete time Markov chain} (DTMC) $\dtmc$ is a tuple $\dtmc = (\mstates, \minit, \mtransitions, \mlabelling)$ where
        $\mstates$ is a finite set of \emph{states};
        $\minit \in \mstates$ is the \emph{initial state};
        $\mtransitions \colon \mstates \times \mstates \to \interval{0}{1}$ is a \emph{transition function} such that for each $s \in \mstates$, we have $\sum_{s' \in \mstates} \mtransitions(s, s') = 1$; 
        and
        $\mlabelling \colon \mstates \to 2^{\AP}$ is a \emph{labelling function}.
\end{definition}
The \emph{underlying graph} of a DTMC $\dtmc = (\mstates, \minit, \mtransitions, \mlabelling)$ is a directed graph $\langle V, E \rangle$ with $V = \mstates$ as vertexes and $E = \setcond{(s, s') \in \mstates \times \mstates}{\mtransitions(s, s') > 0}$ as edges.

\begin{figure}[t]
    \centering
    \includegraphics{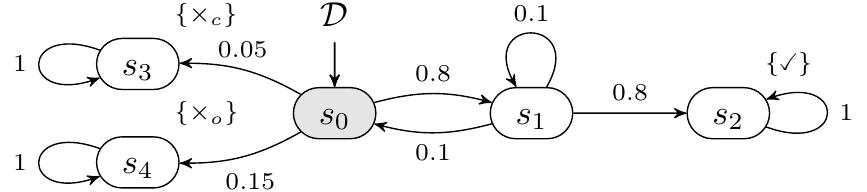}
    \caption{An example of discrete time Markov chain}
    \label{fig:dtmc}
\end{figure}
As an example of DTMC, consider the DTMC $\dtmc$ shown in Fig.~\ref{fig:dtmc}.
$\dtmc$ has 5 states (from $s_{0}$ to $s_{4}$), with $s_{0}$ being the initial one (marked with the gray background and the small incoming arrow);
transitions with probability larger than $0$ are depicted as arrows, so for example we have $\mtransitions(s_{0}, s_{1}) = 0.8 > 0$, while the labels assigned to each state are shown on the top-right corner of the state itself, e.g., $\mlabelling(s_{2}) = \setnocond{\success}$ while $\mlabelling(s_{0}) = \emptyset$.

DTMCs can be equipped with reward structures that assign values to states and transitions;
such reward structures can be used to count the number of transitions taken so far or to attach ``costs'' or ``gains'' to the DTMC.
\begin{definition}
\label{def:DTMRM}
    A \emph{discrete time Markov reward model} (DTMRM) $\dtmrm$ is a pair $\dtmrm = (\dtmc, \mreward)$ where $\dtmc$ is a DTMC and $\mreward \colon \mstates \cup (\mstates \times \mstates) \to \posreals$ is a \emph{reward function}.
\end{definition}
For example, the reward function $\mreward[c]$ defined as $\mreward[c](s) = 0$ and $\mreward[c](s, s') = 1$ for each $s, s' \in \mstates$ allows us to ``count'' the number of steps taken by the DTMC.

Let $\dtmc$ be a DTMC;
a \emph{path} $\mpath$ of $\dtmc$ is a (possibly infinite) sequence of states $\mpath = s_{0} s_{1} s_{2} \cdots$ such that for each meaningful $i \in \naturals$, we have $\mtransitions(s_{i}, s_{i+1}) > 0$; 
we write $\mpath_{i}$ to indicate the state $s_{i}$.
We let $\mpathsfin{\dtmc}$ and $\mpaths{\dtmc}$ denote the sets of all finite and infinite paths of $\dtmc$, respectively.
Given a finite path $\mpath = s_{0} s_{1} s_{2} \cdots s_{n}$, we denote by $\size{\mpath}$ the number of states $n + 1$ of $\mpath$.

Given a finite path $\mpath$, the \emph{cylinder} of $\mpath$, denoted by $\cylinder{\mpath}$, is the set of infinite paths having $\mpath$ as prefix. 
Given a state $s \in \mstates$, we define the probability of the cylinder set of $\mpath$ by $\prob_{s}^{\dtmc}\big(\cylinder{\mpath}\big) = \dirac{s}(\mpath_{0}) \cdot \prod_{i = 0}^{\size{\mpath} - 1} \mtransitions(\mpath_{i}, \mpath_{i+1})$, where $\dirac{s}(s')$ is $1$ if $s' = s$ and $0$ otherwise. 
For a given DTMC $\dtmc$, $\prob_{s}^{\dtmc}$ can be uniquely extended to a probability measure over the $\sigma$-algebra generated by all cylinder sets; 
see~\cite{DBLP:books/BaierK08PrinciplesMC} for more details.
In the remainder of the paper, we might just write $\prob_{s}$ instead of $\prob_{s}^{\dtmc}$ when $\dtmc$ is clear from the context.

Given a DTMRM $\dtmrm = (\dtmc, \mreward)$, similarly to $\prob_{s}^{\dtmc}$ we can define the \emph{expected cumulative reward} $\exprew_{s}^{\dtmrm}$ as follows (cf.~\cite{DBLP:books/BaierK08PrinciplesMC,DBLP:journals/sttt/HahnHZ11,DBLP:conf/sfm/KwiatkowskaNP07}):
given set $T \subseteq \mstates$ of states, $\exprew_{s}^{\dtmrm}(T)$ is the expectation of the random variable $X^{T} \colon \mpaths{\dtmc} \to \posreals$ with respect to the probability measure $\prob_{s}^{\dtmc}$ defined as follows:
\[
    X^{T}(\mpath) = 
    \begin{cases}
        0 & \text{if $\mpath_{0} \in T$,} \\
        \infty & \text{if $\mpath_{i} \notin T$ for each $i \in \naturals$,} \\
        \sum_{i = 0}^{\min\setcond{n \in \naturals}{\mpath_{n} \in T} - 1} \mreward(\mpath_{i}) + \mreward(\mpath_{i}, \mpath_{i + 1}) & \text{otherwise.}
    \end{cases}
\]

\subsection{Probabilistic Reward Logic PRCTL}
\label{ssec:logicPRCTL}

To express properties about probabilistic models with rewards, we use formulas from PRCTL, the Probabilistic Reward CTL logic~\cite{DBLP:conf/formats/AndovaHK03}, that extends PCTL~\cite{DBLP:journals/fac/HanssonJ94,DBLP:conf/fsttcs/BiancoA95} with rewards.
Such formulas are constructed according to the following grammar, where $\lsf$ is a \emph{state formula} and $\lpf$ is a \emph{path formula}:
\begin{align*}
    \lsf & ::= a \mid \lnot \lsf \mid \lsf \land \lsf \mid \lP{\mathord{\bowtie p}}(\lpf) \mid \lER{\mathord{\bowtie r}}(\lF \lsf) \\
    \lpf & ::= \lX \lsf \mid \lsf \lU \lsf \mid \lsf \lBU{k} \lsf
\end{align*}
where $a \in \AP$, 
$\mathord{\bowtie} \in \setnocond{\mathord{<}, \mathord{\leq}, \mathord{\geq}, \mathord{>}}$, 
$p \in \interval{0}{1}$, 
$r \in \posreals$,
and $k \in \naturals$.
We use freely the usually derived operators, like $\lsf_{1} \lor \lsf_{2} = \lnot(\lnot \lsf_{1} \land \lnot\lsf_{2})$, $\ltrue = a \lor \neg a$, and $\lF \lsf = \ltrue \lU \lsf$.
The PCTL logic is just PRCTL without the $\lER{\mathord{\bowtie r}}(\lF \lsf)$ operator.

The semantics of a state formula $\lsf$ and of a path formula $\lpf$ is given with respect to a state $s$ and a path $\mpath$ of a DTMRM $\dtmrm = (\dtmc, \mreward)$, respectively.
The semantics is standard for all Boolean and temporal operators (see, e.g.,~\cite{DBLP:books/BaierK08PrinciplesMC,DBLP:reference/mc/2018}); 
for the $\lP{\mathord{\bowtie} p}$ operator, it is defined as $s \models \lP{\mathord{\bowtie} p}(\lpf)$ iff $\prob_{s}(\setcond{\mpath \in \mpaths{\dtmc}}{\mpath \models \lpf}) \bowtie p$ and, similarly, $s \models \lER{\mathord{\bowtie} r}(\lpf)$ iff $\exprew_{s}(\setcond{\mpath \in \mpaths{\dtmc}}{\mpath \models \lpf}) \bowtie r$.

With some abuse of notation, we write $\dtmrm \models \lsf$ if $\minit \models \lsf$; 
we also consider $\lP{\mathord{=}?}(\lpf)$ and  $\lER{\mathord{=}?}(\lpf)$ as PRCTL formulas, asking to compute the probability (resp.\@ expected reward) of satisfying $\lpf$ in the initial state $\minit$ of $\dtmrm$, i.e., to compute the value $\prob_{\minit}(\setcond{\mpath \in \mpaths{\dtmc}}{\mpath \models \lpf})$ (resp.\@ $\exprew_{\minit}(\setcond{\mpath \in \mpaths{\dtmc}}{\mpath \models \lpf})$).

Consider the DTMC $\dtmc$ shown in Fig.~\ref{fig:dtmc}. 
As an example of PRCTL formula, there is $\lP{\mathord{=}?}(\lF \success)$ that asks to compute the probability of eventually reaching a state labelled with $\success$, for which we have $\lP{\mathord{=}?}(\lF \success) \approx 0.78$.

\subsection{Parametric Models}
\label{ssec:parametricModels}

We now recall the definition of parametric models from~\cite{DBLP:conf/nfm/HahnHZ11,DBLP:journals/sttt/HahnHZ11}.
Given a finite set of \emph{variables}, or \emph{parameters}, $\parametersSet = \setnocond{\parameter_{1}, \dotsc, \parameter_{n}}$, let $\parameters = (\parameter_{1}, \dotsc, \parameter_{n})$ denote the vector of parameters and $\range \colon \parametersSet \to \reals$ be the function assigning to each parameter $\parameter \in \parametersSet$ its closed interval $\range(\parameter) = \interval{L_{\parameter}}{U_{\parameter}} \subseteq \reals$ of valid values.
Given the field $\polynomialsSet$ of the polynomials with variables $\parametersSet$, a \emph{rational function} $f$ is a fraction $f(\parameters) = \frac{g_{1}(\parameters)}{g_{2}(\parameters)}$ where $g_{1}, g_{2} \in \polynomialsSet$;
let $\ratfunsSet$ denote the set of rational functions.
An \emph{evaluation} $\evaluation$ is a function $\evaluation \colon \parametersSet \to \reals$ such that for each $\parameter \in \parametersSet$, $\evaluation(\parameter) \in \range(\parameter)$.
Given $f = \frac{g_{1}}{g_{2}} \in \ratfunsSet$ and an evaluation $\evaluation$, we denote by $\evaluate{f}$ the rational number $f(\evaluation(\parameters)) = f(\evaluation(\parameter_{1}), \dotsc, \evaluation(\parameter_{n}))$;
we assume that $\evaluate{f}$ is well defined for each evaluation $\evaluation$, that is, $\evaluate{g_{2}} \neq 0$ for each evaluation $\evaluation$.

\begin{definition}
\label{def:pdtmc}
    Given a finite set of parameters $\parametersSet$, a \emph{parametric discrete time Markov chain} (pDTMC) $\pdtmc$ with parameters $\parametersSet$ is a tuple $\pdtmc = (\mstates, \minit, \mtransitions, \mlabelling)$ where $\mstates$, $\minit$, and $\mlabelling$ are as in Def.~\ref{def:dtmc}, while $\mtransitions \colon \mstates \times \mstates \to \ratfunsSet$.
\end{definition}

\begin{definition}
\label{def:pdtmcInstantiation}
    Given a pDTMC $\pdtmc = (\mstates, \minit, \mtransitions, \mlabelling)$, an evaluation $\evaluation$ \emph{induces} the DTMC $\evaluate{\dtmc} = (\mstates, \minit, \mtransitions_{\evaluation}, \mlabelling)$, provided that $\mtransitions_{\evaluation}(s, s') = \evaluate{\mtransitions(s, s')}$ for each $s, s' \in \mstates$ satisfies the conditions given in Def.~\ref{def:dtmc}.
\end{definition}

The extension to parametric DTMRMs (pDTMRMs) is trivial:
a pDTMRM $\pdtmrm$ is just a pair $\pdtmrm = (\pdtmc, \mreward)$ where $\pdtmc$ is a pDTMC and $\mreward$ is a reward function.

To simplify the presentation and ensure that the underlying graph of $\pdtmc$ does not depend on the actual evaluation, we make the following assumption:
\begin{assumption}[cf.~\cite{DBLP:conf/nfm/HahnHZ11}]
\label{asmt:sameGraph}
    Given a pDTMC $\pdtmc$, for each pair of evaluations $\evaluation_{1}$ and $\evaluation_{2}$, for the induced DTMCs $\evaluate[\evaluation_{1}]{\pdtmc}$ and $\evaluate[\evaluation_{2}]{\pdtmc}$ we have that for each $s, s' \in \mstates$, it holds that $\mtransitions_{\evaluation_{1}}(s, s') = 0$ if and only if $\mtransitions_{\evaluation_{2}}(s, s') = 0$.
\end{assumption}
By this assumption, either a state $s'$ has probability $0$ to be reached from $s$ (i.e., it is not reachable) independently of the evaluation, or it is always reachable, with possibly different probability values.

\begin{figure}[t]
    \centering
    \includegraphics{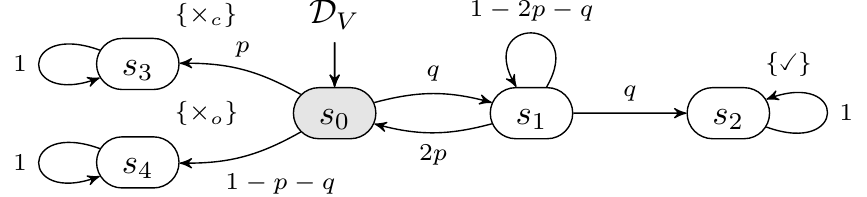}
    \caption{An example of parametric discrete time Markov chain}
    \label{fig:pdtmc}
\end{figure}
As an example of pDTMC, consider the model shown in Fig.~\ref{fig:pdtmc}:
now, $p$ and $q$ are the parameters, with e.g.\@ $\range(p) = \interval{0.01}{0.09}$ and $\range(q) = \interval{0.25}{0.8}$.
One evaluation is $\evaluation(p) = 0.05$ and $\evaluation(q) = 0.8$, which gives us the DTMC shown in Fig.~\ref{fig:dtmc}.
The rational function corresponding to the PRCTL formula $\lP{\mathord{=}?}(\lF \success)$ is $\frac{q^{2}}{q + 2p - 2pq} \approx 0.78$ when evaluated on $\evaluation$, as one would expect.

\section{Probably Approximately Correct Function Synthesis}
\label{sec:synthesisPACfunctions}

In this section, we show how to approximate the actual functions with low-degree polynomials, while providing a statistical PAC guarantee on the closeness of the approximating polynomial with the approximated function.

\subsection{Probably Approximately Correct Models}
\label{ssec:PACmodels}

Our  method provides a PAC approximation, with respect to the given significance level $\significanceLevel$ and error rate $\errorRate$. 
First, we define the PAC approximation of a generic function $f$ as follows.
\begin{definition}
\label{def:PACmodel}
    Given a set of $n$ variables $\parametersSet$, their domain $X = \prod_{i = 1}^{n} \range(\parameter_{i})$, and a function $f \colon X \to \reals$, let $P$ be a probability measure over $X$, $\margin \in \posreals$ be a margin to measure the approximation error, and $\errorRate$, $\significanceLevel \in (0,1]$ be an error rate and a significance level, respectively.
    
    We say that the polynomial $\approxfun{f} \in \polynomialsSet$ is a PAC approximation of $f$ with $(\errorRate, \significanceLevel)$-guarantee if, with confidence $1 - \significanceLevel$, the following condition holds:
    \[
        P(\abs{\approxfun{f}(\parameters) - f(\parameters)} \leq \margin) \geq 1 - \errorRate.
    \]
\end{definition}

In this work, we assume that $P$ is the uniform distribution on the domain $X = \prod_{i = 1}^{n} \range(\parameter_{i})$ unless otherwise specified. 
Intuitively, our aim is to make the PAC approximation $\approxfun{f}$ as close as possible to $f$, so we introduce the margin $\margin$ to describe how close the two functions are. 
The two statistical parameters $\significanceLevel$ and $\errorRate$ are the significance level and error rate, respectively; 
they are used to measure how often the difference between $\approxfun{f}$ and $f$ respects the threshold $\margin$, so we can adjust these parameters to change the quality of the approximation.

\subsection{The Scenario Approach}
\label{ssec:scenarioapproach}

PAC approximation is inspired by the scenario approach proposed in~\cite{DBLP:journals/tac/CalafioreC06,DBLP:journals/arc/CampiGP09}. 
We consider the following class of convex optimization problems: 
\begin{equation}
\label{eq:optProblemExact}
    \begin{split}
        \min \limits_{\theta \in \Theta \subseteq \reals^{m}} & \quad \vec{a}^{T} \theta \\
        \mathrm{s.t.} & \quad f_{\omega}(\theta) \leq 0 \qquad  \forall \omega \in \Omega
    \end{split}
\end{equation}
under the assumption that $f_{\omega} \colon \Theta \to \reals$ is a convex function of $\theta \in \Theta$ for every $\omega \in \Omega$.
Moreover, we assume that the domains $\Theta$ and $\Omega$ are convex and closed.

The main obstacle on solving the optimization problem~\eqref{eq:optProblemExact} is that in general it has infinitely many constraints, due to the convexity of $\Omega$. 
Instead of solving the problem~\eqref{eq:optProblemExact}, it was proposed in~\cite{DBLP:journals/tac/CalafioreC06} to use finitely many sampled points  that provide statistical guarantee on the error rate made with respect to the exact solution of the problem~\eqref{eq:optProblemExact}, which is formalized as follows.
\begin{definition}
\label{def:scenarioApproach}
    Given a convex and closed set $\Omega$ and a constant $l \in \naturals$, let $P$ be a probability measure over $\Omega$ and $\omega_{1}, \dotsc, \omega_{l}$ be $l$ independent identically distributed samples taken from $\Omega$ according to $P$.
    The \emph{scenario design problem} corresponding to the problem~\eqref{eq:optProblemExact} is defined as
    \begin{equation}
    \label{eq:optProblemPAC}
        \begin{split}
            \min\limits_{\theta \in \Theta \subseteq \reals^{m}} & \quad \vec{a}^{T} \theta \\
            \mathrm{s.t.} & \quad \bigwedge_{i=1}^{l} f_{\omega_{i}}(\theta) \leq 0 \qquad \omega_{i} \in \Omega
        \end{split}
    \end{equation}
\end{definition}

The optimization problem~\eqref{eq:optProblemPAC} can be seen as the relaxation of the optimization problem~\eqref{eq:optProblemExact}, since we do not require that the solution $\theta^{*}_{l}$ of the problem~\eqref{eq:optProblemPAC} satisfies all constraints $f_{\omega}(\theta^{*}_{l}) \leq 0$ for each $\omega \in \Omega$, but only the constraints corresponding to the $l$ samples from $\Omega$ according to $P$. 
The issue now is how to provide enough guarantee that the optimal solution $\theta^{*}_{l}$ of~\eqref{eq:optProblemPAC} also satisfies the other constraints $f_{\omega}(\theta) \leq 0$ with $\omega \in \Omega \setminus \setnocond{\omega_{i}}_{i=1}^{l}$ we have not considered.

To answer this question, an \emph{error rate} $\errorRate$ is introduced to bound the probability that the solution $\theta^{*}_{l}$ violates the constraints of problem~\eqref{eq:optProblemExact}; 
we denote by $\significanceLevel$ the \textit{significance level} with respect to the random sampling solution algorithm.
Statistics theory ensures that as the number of samples $l$ increases, the probability that the optimal solution of the optimization problem~\eqref{eq:optProblemPAC} violates the other unseen constraints will tend to zero rapidly.
The minimal number of sampled points $l$ is related to the error rate $\errorRate \in (0, 1]$ and significance level $\significanceLevel \in (0, 1]$ by: 
\begin{theorem}[\kern-0.9ex\cite{DBLP:journals/arc/CampiGP09}]
\label{thm:PACnumberOfSamples}
    If the optimization problem~\eqref{eq:optProblemPAC} is feasible and has a unique optimal solution $\theta^{*}_{l}$, then $P(f_{\omega}(\theta^{*}_{l})>0)<\errorRate$, with confidence at least $1 - \significanceLevel$, provided that the number of constraints $l$ satisfies
    \[
        l \geq \frac{2}{\errorRate} \cdot \Big(\ln \frac{1}{\significanceLevel} + m \Big),
    \]
    where $m$ is the dimension of $\theta$, that is, $\theta \in \Theta \subseteq \reals^{m}$, $\errorRate$ and $\significanceLevel$ are the given error rate and significance level, respectively.
\end{theorem}

In Theorem~\ref{thm:PACnumberOfSamples}, we assume that the optimization problem~\eqref{eq:optProblemPAC} has a unique optimal solution $\theta^{*}_{l}$.
This is not a restriction in general, since for multiple optimal solutions we can just use the Tie-break
rule~\cite{DBLP:journals/tac/CalafioreC06} to get a unique optimal solution.

\subsection{Synthesizing Parametric Functions}
\label{ssec:ourMethod}

We now  apply the above scenario approach to the synthesis of the parametric functions for  pDTMRMs. 
Given a pDTMRM $\pdtmrm = (\pdtmc, \mreward)$ with $\pdtmc = (\mstates, \minit, \mtransitions, \mlabelling)$, let $\parameters$ denote the vector of parameters $(\parameter_{1}, \dotsc, \parameter_{n})$ of $\pdtmc$. 
For a PRCTL state formula $\lsf$, the analytic function $f_{\lsf}(\parameters)$, representing the probability or the expected reward of the paths satisfying $\lsf$ in the pDTMRM $\pdtmrm$, can be a rational function with a very complicated form~\cite{DBLP:conf/cav/HahnHWZ10,DBLP:journals/sttt/HahnHZ11}. 
Our aim is to approximate the function $f_{\lsf}(\parameters)$ with some low degree polynomial $\approxfun{f}_{\lsf}(\parameters)$, such as a quadratic polynomial $\approxfun{f}_{\lsf}(\parameters) = \vec{c}_{0} + \vec{c}_{1} \cdot \parameters + \vec{c}_{2} \cdot \parameters^{2} = (\vec{c}_{0}, \vec{c}_{1}, \vec{c}_{2}) \cdot (1, \parameters, \parameters^{2})^{T}$. 

The reason why we choose a polynomial $\approxfun{f}_{\lsf}(\parameters)$ with low degree to fit the rational function $f_{\lsf}(\parameters)$ is that the graph of polynomials $\approxfun{f}_{\lsf}(\parameters)$ and original functions $f_{\lsf}(\parameters)$ are both surfaces and the polynomial $\approxfun{f}_{\lsf}(\parameters)$ can approximate the rational function $f_{\lsf}(\parameters)$ well if we synthesize appropriately the coefficients $\vec{c} = (\vec{c}_{0}, \vec{c}_{1}, \vec{c}_{2})$ of the polynomial by learning them. 

It is worth mentioning that no matter how complicated the function $f_{\lsf}(\parameters)$ is (it could also be any kind of function other than rational functions), we can still obtain an approximating polynomial $\approxfun{f}_{\lsf}(\parameters)$ of $f_{\lsf}(\parameters)$ by solving an optimization problem, and utilize it to analyze various properties the original function $f_{\lsf}(\parameters)$ may satisfy. 
In the remainder of this section, we show how we synthesize such coefficients $\vec{c}$, and thus the polynomial; 
we first introduce some notations.

Given the vector of parameters $\parameters$ and a degree $d \in \naturals$, we denote by $\parameters^{d}$ the vector of monomials $\parameters^{d} = (\parameters^{\vec{\alpha}})_{\norm[1]{\vec{\alpha}} = d}$, where each monomial $\parameters^{\vec{\alpha}}$ is defined as $\parameters^{\vec{\alpha}} = \parameter_{1}^{\alpha_{1}} \parameter_{2}^{\alpha_{2}} \cdots \parameter_{n}^{\alpha_{n}}$, with $\vec{\alpha} = (\alpha_{1}, \dotsc, \alpha_{n}) \in \naturals^{n}$ and $\norm[1]{\vec{\alpha}} = \sum_{i = 1}^{n} \alpha_{i}$.
Then, we associate a coefficient $\vec{c}_{i}$ to each of the monomials in the vector $(\parameters^{i})_{i = 0}^{d}$, obtaining the PAC approximation $\approxfun{f}(\parameters) = \sum_{i = 0}^{d} \vec{c}_{i} \cdot \parameters^{i}$.
For example, if the pDTMC $\pdtmc$ has two parameters $\parameter_{1}$ and $\parameter_{2}$, then for $d = 2$ we get the quadratic polynomial $\approxfun{f}(\parameters) = \vec{c}_{0} + \vec{c}_{1} \cdot \parameters + \vec{c}_{2} \cdot \parameters^{2} = c_{0} + (c_{11} \cdot \parameter_{1} + c_{12} \cdot \parameter_{2}) + (c_{21} \cdot \parameter_{1}^{2} + c_{22} \cdot \parameter_{1} \cdot \parameter_{2} + c_{23} \cdot \parameter_{2}^{2} )$.
In general, for $n$ parameters and a polynomial of degree $d$, we need $\binom{n + d}{n}$ coefficients.

Given the PAC approximation schema $\approxfun{f}(\parameters) = \sum_{i = 0}^{d} \vec{c}_{i} \cdot \parameters^{i} = \vec{c}\cdot (1,\parameters,\cdots,\parameters^{d})^{T}$, we solve the following Linear Programming (LP) problem to learn the coefficients $\vec{c}$ of the polynomial $\approxfun{f}(\parameters)$: 
\begin{equation}
\label{eq:LPpolynomialApproximation}
\begin{split}
	\min\limits_{\vec{c}, \margin} & \quad \margin \\
	\mathrm{s.t.} & \quad -\margin \leq f(\parameters) - \vec{c} \cdot (1, \parameters, \dotsc, \parameters^{d})^{T} \leq \margin, \qquad \forall \parameters \in X,\\
	& \quad \vec{c} \in \reals^{\binom{n + d}{n}}, \margin \geq 0
\end{split}
\end{equation}
where $f(\parameters)$ is the analytic function on the domain $X = \prod_{i = 1}^{n} \range(\parameter_{i})$. 
Note that for pDTMRMs we do not need to compute the rational function $f_{\lsf}$ used as $f$ in problem~\eqref{eq:LPpolynomialApproximation} to get its value on $\parameters$, since we can first instantiate the pDTMRM with $\parameters$ and then compute the value of $\lsf$ in the instantiated~DTMRM.

Given the error rate $\errorRate$ and the significance level $\significanceLevel$, by Theorem~\ref{thm:PACnumberOfSamples} we need only to independently and identically sample at least $l \geq \frac{2}{\errorRate}\big(\ln\frac{1}{\significanceLevel} + \binom{n + d}{n} + 1\big)$ points $\approxfun{X} = \{\parameters_{i}\}_{i=1}^{l}$ to form the constraints used in the relaxed LP problem, as done in the problem~\eqref{eq:optProblemPAC}.
Concretely, we get the following LP problem:
\begin{equation}
\label{eq:PACLPpolynomialApproximation}
\begin{split}
	\min\limits_{\vec{c}, \margin} & \quad \margin \\
	\mathrm{s.t.} &  \quad \bigwedge^{l}_{i=1} -\margin \leq f(\parameters_{i}) - \vec{c} \cdot (1,\parameters_{i},\cdots,\parameters^{d}_{i})^{T} \leq \lambda, \qquad \forall \parameters_{i} \in \approxfun{X},  \\
	& \quad \vec{c} \in \reals^{\binom{n + d}{d}}, \margin \geq 0.
\end{split}
\end{equation}
We solve the optimization problem~\eqref{eq:PACLPpolynomialApproximation} to get the coefficients $\vec{c}$, hence the PAC approximation $\approxfun{f}$ of the original function $f$, with the statistical guarantees given by Def.~\ref{def:PACmodel}; 
in the context of a pDTMRM $\pdtmrm$ and a PRCTL state formula $\lsf$, we get the PAC approximation $\approxfun{f}_{\lsf}$ of the original function $f_{\lsf}$.

\subsection{PRCTL Property Analysis}
\label{subsec:reachability}

Given the probabilistic formula $\lsf = \lP{=?}(\lpf)$ with path formula $\lpf$,  we can obviously use the PAC approximation $\approxfun{f}_{\lsf}$ to check whether the domain of parameters $X$ is safe, with PAC guarantee.
In this section, we introduce a direct PAC based approach for checking domain's safety, without having to learn the approximations first. 
Then, we consider linear approximations and discuss how counterexamples can be generated in this case before showing how the polynomial PAC approximation $\approxfun{f}_{\lsf}$ can be used to analyze global properties of $f_{\lsf}$ over the whole parameter space $X$.
Lastly, we present how to extend the approach to the reward formula $\lsf = \lER{=?}(\lF \lsf')$.

\begin{definition}[Safe Region]
\label{def:safe region}
    Let $X = \prod_{i = 1}^{n} \range(\parameter_{i})$ be the domain of a set of $n$ parameters $\parametersSet$.
    Given a function $f \colon X \to \posreals$ and a safety level $\safetyLevel \in \posreals$, we say that the point $\parameters \in \parametersSet$ is \emph{safe} if and only if $f(\parameters) < \safetyLevel$; 
    we call $X$ safe if and only if each $\parameters \in \parametersSet$ is safe.
\end{definition}

Intuitively, we hope that the probability of the pDTMRM $\pdtmrm$ to reach an unsafe state under any choice of the parameters will be less than the given safety level, which is the motivation for defining the safe region.
To check whether the domain $X$ of the parameters is safe, we can resort to solve the following optimization problem with respect to the given error rate $\errorRate$ and significance level $\significanceLevel$, and compare the obtained optimal solution $\margin^{*}$ with $\safetyLevel$:
\begin{equation}
\label{eq:SafeRegion}
\begin{split}
	\min & \quad \margin \\
	\mathrm{s.t.} & \quad f(\parameters) \leq \margin \qquad \forall \parameters \in \approxfun{X},
\end{split}
\end{equation}
where $\approxfun{X} \subseteq X$ is a set of samples such that $\size{\approxfun{X}} \geq \left\lceil \frac{2}{\errorRate} \cdot (\ln \frac{1}{\significanceLevel} + 1) \right\rceil$.
The optimization problem~\eqref{eq:SafeRegion} can be solved in time $\bigO(\size{\approxfun{X}})$, since it only needs to compute the maximum value of $f_{\lsf}(\parameters)$ for $\parameters \in \approxfun{X}$ as the optimal solution $\margin^{*}$.  
Although the calculation is very simple, polynomials with degree 0, i.e., constants, also have good probability and statistical meaning, so we have the following result as a direct consequence of the definitions:
\begin{restatable}{lemma}{lemZeroApproximation}
\label{lem:zeroApproximation}
    Given the safety level $\safetyLevel$, if the optimal solution $\margin^{*}$ of the problem~\eqref{eq:SafeRegion} satisfies $\margin^{*} < \safetyLevel$, then the domain $X$ is safe with $(\errorRate, \significanceLevel)$-guarantee.
    Otherwise, if $\margin^{*} \geq \safetyLevel$, then the parameter point $\parameters^{*} \in \approxfun{X}$ corresponding to $\margin^{*}$ is unsafe.
\end{restatable}

By Lemma~\ref{lem:zeroApproximation}, we can analyze with $(\errorRate, \significanceLevel)$-guarantee whether the parameter space is safe or not. 
For example, consider the pDTMC $\pdtmc$ shown in Fig.~\ref{fig:pdtmc} and the safety property $\lP{< 0.8}(\lF (\failure_{c} \lor \failure_{o}))$. 
If we set $\errorRate = \significanceLevel = 0.05$, by sampling in the region $X = \interval{0.01}{0.09} \times \interval{0.25}{0.8}$ at least 160 points and solving the resulting optimization problem~\eqref{eq:SafeRegion}, we get the optimal value $\margin^{*} = 0.747$ by rounding to three decimals. 
Since $\margin^{*} = 0.747 < 0.8$, by Lemma~\ref{lem:zeroApproximation}, the region $X$ is safe with $(0.05, 0.05)$-guarantee.

\subsubsection{Linear PAC Approximation and Counterexamples.}
\label{sssec:1degree}

Since constants can approximate the maximum value of the function $f$ with the given $(\errorRate,\significanceLevel)$-PAC guarantee, linear functions can also be used to approximate $f$, which are more precise than constants.
Also, we can check whether there is an unsafe region in the domain of parameters $X$ with a given confidence, by the following Lemma~\ref{lem:linearApproximation}, and further search counterexamples by linear PAC approximations.
\begin{restatable}{lemma}{lemLinearApproximation}
\label{lem:linearApproximation}
    Given the domain of parameters $X$, a function $f \colon X \to \posreals$, and a probability measure $P$ over $X$, let $\approxfun{f}$ be a PAC approximation of $f$ with $(\errorRate, \significanceLevel)$-guarantee.
    Given the safety level $\safetyLevel \in \posreals$, if for each $\parameters \in X$ we have $\approxfun{f}(\parameters) + \margin < \safetyLevel$, then $P(f(\parameters) < \safetyLevel) \geq 1 - \errorRate$ holds with confidence $1 - \significanceLevel$. 
    In turn, if $P(\approxfun{f}(\parameters) - \margin > \safetyLevel) > \errorRate$, then there exist $\parameters \in X$ such that $f(\parameters) > \safetyLevel$ holds with confidence $1 - \significanceLevel$.
\end{restatable}

\begin{figure}[t]
    \centering
    \resizebox{\linewidth}{!}{
        \includegraphics{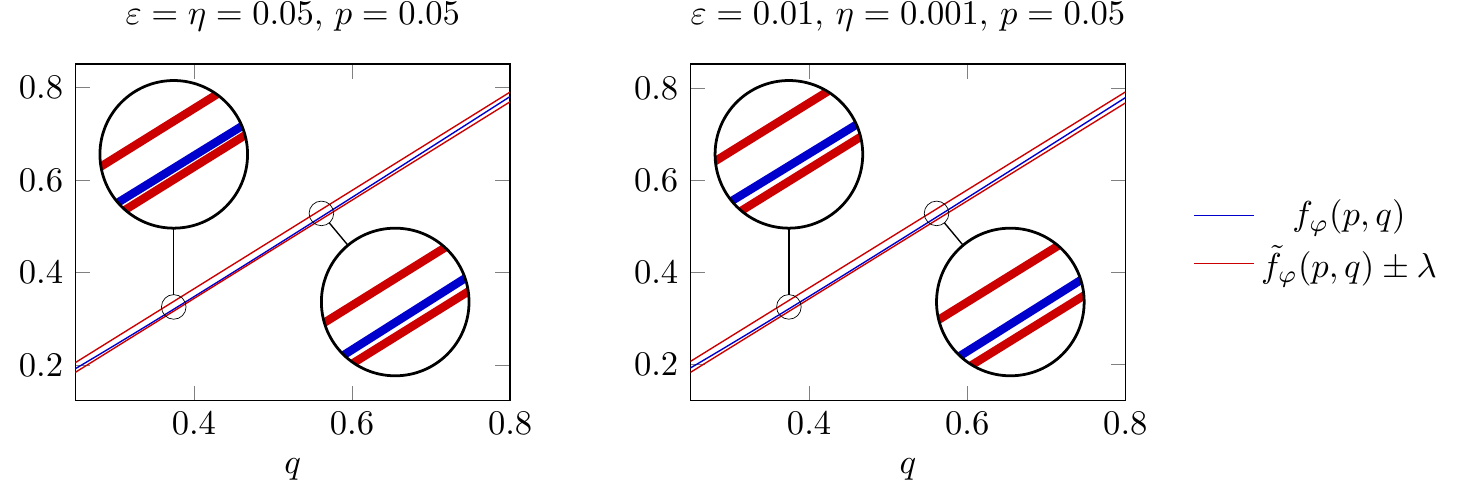}
    }
    \caption{The rational function  $f_{\lsf}(p, q) = \frac{q^{2}}{q + 2p - 2pq}$ and its linear approximations $\approxfun{f}_{\lsf}(p, q)$ with different choices of $\errorRate$ and $\significanceLevel$}
    \label{fig:linearApproximation}
\end{figure}

The plots in Fig.~\ref{fig:linearApproximation} show the results of applying linear PAC approximation on the function $f_{\lsf}(p, q)$, with $\lsf = \lP{=?}(\lF \success)$, for the pDTMC $\pdtmc$ shown in Fig.~\ref{fig:pdtmc}. 
We sampled 280 points for $\errorRate = \significanceLevel = 0.05$ and 2182 points for $\errorRate = 0.01$ and $\significanceLevel = 0.001$, respectively, according to Thm.~\ref{thm:PACnumberOfSamples}. 
The plot on the left, where we fix the parameter $p = 0.05$, shows that even if we sample just 280 points, $f_{\lsf}(p, q)$ and $\approxfun{f}_{\lsf}(p, q)$ are closer than the computed margin $\margin$. 
For the case $\errorRate = \significanceLevel = 0.05$, the linear approximation is $\approxfun{f}_{\lsf}(p, q) = -0.035 + 1.063 * q - 0.718 * p$ with $\margin = 0.011$ by rounding the coefficients to three decimals. 
We can easily check that for each $(p,q) \in X$ we have $\approxfun{f}_{\lsf}(p,q) + \margin < 0.85$ by linear programming, so $X = \interval{0.01}{0.09} \times \interval{0.25}{0.8}$ is a $0.85$-safe region with respect to $f_{\lsf}(p,q)$ with $(0.05,0.05)$-guarantee. 
However, if we set $\safetyLevel = 0.6$, we can prove $P(\approxfun{f}_{\lsf} - \margin > \safetyLevel) = 0.288 > \errorRate = 0.05$, so by Lemma~\ref{lem:linearApproximation} we get that there exist an unsafe region such that $f(p,q) > \safetyLevel$, with confidence $95\%$.

We can take advantage of the easy computation of linear programming with linear functions to further search for potential counterexamples that may exist.
The maximum value of $\approxfun{f}_{\lsf}$ can be found at $(0.01, 0.8)$, according to the linearity of $\approxfun{f}_{\lsf}$, so we can instantiate the pDTMC $\pdtmc$ in Fig.~\ref{fig:pdtmc} with the parameter point $(0.01, 0.8)$ to get that $f_{\lsf}(p,q) = 0.796$. 
Since $f_{\lsf}(p,q) > 0.6$ for the safety level $\safetyLevel = 0.6$, we can claim that the \emph{real counterexample} $(0.01, 0.8)$ is found.
In the case that the parameter point $\parameters_{0} = (p,q)$ corresponding to maximum value of $\approxfun{f}_{\lsf}$ is a spurious counterexample for the pDTMC with respect to $\lsf$, we can learn a more precise approximation by adding $\parameters_{0}$ to $\approxfun{X}$. 
One may also divide the domain $X$ into several subdomains and analyze each of them separately. 

As for the computational complexity, it is easy to find the maximum value of a linear function by linear programming; 
on the other hand, computing the maximum value of polynomials and rational functions is rather difficult if their degree is very high or the dimension of the parameter space is too large.  
So a linear function is a good alternative to compute the maximum value of $f$ with PAC guarantee, while polynomials are suitable for analyzing more complicated properties, such as the global ones considered below.

\subsubsection{Polynomial PAC Approximation.}
\label{sssec:polynomialApproximation}

One advantage of polynomials over rational functions is that they make it easy to compute complex operations such as inner product and integral~\cite{rudin1976principles}, as needed to evaluate e.g.\@ the $L_{p}$ norm $\norm[p]{g} = \sqrt[p]{\int_{Z} \abs{g(z)}^{p}\,dz}$ of a function $g \colon Z \to \reals$, with $p \geq 1$. 
This means that we can adopt polynomials to check some more complicated properties of a pDTMRM $\pdtmrm$, such as whether the function $f_{\lsf}$ is close to a given number $\beta$ on the whole parameter space $X$. 
This is useful, for instance, to evaluate how much the behavior of $\pdtmrm$ with respect to the property $\lsf$ is affected by the variations of the parameters.
We can model this situation as follows:
\begin{definition}
\label{def:fluctuations}
    Given the domain $X$ of a set of parameters, a function $f \colon X \to \posreals$, a safety level $\safetyLevel$, and $\beta \in \posreals$, we say that \emph{$f$ is near $\beta$ within the safety level $\safetyLevel$} on $X$ with respect to the $L_{p}$ norm, if $\norm[p]{f - \beta} < \safetyLevel$.
\end{definition}
To verify the above property, we can rely on the following result:
\begin{restatable}{lemma}{lemNearMargin}
\label{lem:nearMargin}
    Given $X$, $f$, $\safetyLevel$, and $\beta$ as in Def.~\ref{def:fluctuations}, let $M$ be an upper bound of $f(X)$ and $\approxfun{f}$ be a PAC approximation of $f$ with $(\errorRate, \significanceLevel)$-guarantee and margin $\margin$; 
    let $\size{X} = \int_{X} 1\,d\parameters$.
    For each $p \geq 1$, if $\approxfun{f}$ satisfies the condition 
    \begin{equation}
    \label{eq:polyUpperbound}
        \sqrt[p]{ \left(\margin \sqrt[p]{(1 - \errorRate) \cdot \size{X} } + \norm[p]{\approxfun{f} - \beta}\right)^{p} + \errorRate \cdot \size{X} \cdot \max(\size{M-\beta}^{p}, \beta^{p})} < \safetyLevel
    \end{equation}
    then $\norm[p]{f - \beta} < \safetyLevel$ holds with confidence $1 - \significanceLevel$.
\end{restatable}

Consider again the pDTMC $\pdtmc$ shown in Fig.~\ref{fig:pdtmc} and $\lsf = \lP{=?}(\lF \success)$; 
since $f_{\lsf}$ represents probabilities, we have the well-known upper bound $M = 1$.
Here we consider the $L_{2}$ norm, which is widely used in describing the error between functions in the signal processing field (see, e.g.,~\cite{boggess2015first,conway2019course}), as it can reflect the global approximation properties and is easy to compute. 
To simplify the notation, let $\upperboundFXBeta$ denote the complex expression occurring in the formula~\eqref{eq:polyUpperbound}, that is:
\[
    \upperboundFXBeta(\approxfun{f}_{\lsf}, X, \beta) = \sqrt{ \left(\margin \sqrt{(1 - \errorRate) \cdot \size{X}} + \norm[2]{\approxfun{f}_{\lsf} - \beta}\right)^{2} + \errorRate \cdot \size{X} \cdot \max(\size{1-\beta}^{2}, \beta^{2})}.
\] 
We want to know whether $f_{\lsf}(p, q) = \frac{q^{2}}{q + 2p - 2pq}$ is near $0.5$ within $0.05$, i.e., given the safety level $\safetyLevel = 0.05$, we want to check $\norm[2]{f_{\lsf} - 0.5} < 0.05$.
According to Lemma~\ref{lem:nearMargin}, we first compute a PAC approximation $\approxfun{f}_{\lsf}$ of $f_{\lsf}$. 
By setting $\errorRate = \significanceLevel = 0.05$, we get the quadratic polynomial 
$\approxfun{f}_{\lsf}(p, q) = 0.013 + 0.925 * q  - 1.442 * p  + 0.953 * pq + 2.072 * p^{2} + 0.085 * q^{2}$, by rounding to three decimals. 
In this case, we get $\upperboundFXBeta(\approxfun{f}_{\lsf}, X, \beta) = 0.0432 < \safetyLevel = 0.05$, so Lemma~\ref{lem:nearMargin} applies. 
If, instead, we would have chosen $\safetyLevel' = 0.04$, then we cannot prove $\norm[2]{f_{\lsf} - 0.5} < 0.04$ by relying on Lemma~\ref{lem:nearMargin}. 
To do so, we need to consider the more conservative values $\errorRate = 0.01$ and $\significanceLevel = 0.001$, which give us $\upperboundFXBeta(\approxfun{f}_{\lsf}, X, \beta) = 0.0379 < \safetyLevel' = 0.04$, so we can derive that  $\norm[2]{f_{\lsf} - 0.5} < 0.04$ holds with confidence $99.9\%$.

\subsubsection{Extension to Reward Models.}
\label{sssec:rewardExtension}

The extension of the constructions given above to reward properties is rather easy:
for instance, we can approximate the rational function representing the state property $\lsf = \lER{\mathord{= ?}}(\lF \lsf')$, the reward counterpart of $\lP{\mathord{=}?}(\lpf')$, by instantiating $f_{\lsf}(\parameters_{i})$ in Problem~\eqref{eq:PACLPpolynomialApproximation} with the expected reward value computed on the pDTMC instantiated with $\parameters_{i}$.
Similarly, we can compute linear and polynomial PAC approximations for safe regions, with the latter defined in terms of the value of the reward instead of the probability.

We can consider also the following case:
given a pDTMRM $\pdtmrm$, we want to verify whether the expected value of $\lsf = \lER{=?}(\lF \lsf')$ over the parameters $\parameters$, denoted $f_{\lsf}(\parameters)$, can reach a given reward level $\rewardLevel$. 
This model the  scenarios where, to make a decision, we need to know whether the expectation of the rewards for a certain decision satisfies the given conditions.
We formalize this case as follows:
\begin{definition}
\label{def:fluctuationsAboveRewardLevel}
    Given the domain $X$ of a set of parameters, a function $f \colon X \to \posreals$, a reward level $\rewardLevel$, and a probability measure $P$ over $X$, we say that \emph{the expectation of $f$} on $X$ with respect to $P$ can reach the reward level $\rewardLevel$, if 
    \begin{equation}
    \label{eq:rewardmodel}
        \int_{X} f(\parameters)\,dP(\parameters) > \rewardLevel.
    \end{equation}
\end{definition}

We can resort to the following lemma to check condition~\eqref{eq:rewardmodel}:
\begin{restatable}{lemma}{lemRewardModel}
\label{lem:RewardModel}
    Given $X$, $f$, $P$, and $\rewardLevel$ as in Def.~\ref{def:fluctuationsAboveRewardLevel}, let $\approxfun{f}$ be a PAC approximation of $f$ with $(\errorRate, \significanceLevel)$-guarantee and margin $\margin$. 
    If $\approxfun{f}$ satisfies the condition 
    \begin{equation}
    \label{eq:lowerboundofrewardmodel}
        \int_{X}(\approxfun{f}(\parameters) - \margin)\,dP(\parameters) - \errorRate \cdot \size{X} \cdot \max_{\parameters \in X}(\approxfun{f}(\parameters) - \margin) > \rewardLevel,
    \end{equation}
    then Condition~\eqref{eq:rewardmodel} holds with confidence $1 - \significanceLevel$.
\end{restatable}

\section{Experimental Evaluation}
\label{sec:experiments}

We have implemented the PAC-based analysis approach proposed in Sect.~\ref{sec:synthesisPACfunctions} in a prototype tool \pacmc and evaluated it on several benchmarks:
we considered the DTMCs from the \prism benchmark suite~\cite{DBLP:conf/qest/KwiatkowskaNP12}, and replaced the probabilistic choices in them with parameters.
The probabilistic choices in most of the models correspond to the flip of a fair coin, so we considered three possibles ranges for the parameters, namely $\interval{0.01}{0.33}$, $\interval{0.33}{0.66}$, and $\interval{0.66}{0.99}$, to represent the fact that the coin is strongly unfair to head, rather fair, and strongly unfair to tail, respectively.
For the remaining models, where the choice is managed by the uniform distribution over several outcomes, we split the outcomes into two groups (e.g., odd and even outcomes) and then used a parametric coin and five intervals to choose the group.
By considering the reachability properties available for each DTMC and the choice of the constants controlling the size of the DTMCs, we get a total of 936 benchmarks for our evaluation for probabilistic properties and 620 benchmarks for expected rewards.
We performed our experiments on a desktop machine with an i7-4790 CPU and 16 GB of memory running Ubuntu Server 20.04.4; 
we used \benchexec~\cite{DBLP:journals/sttt/BeyerLW19} to trace and constrain the tools' executions: 
we allowed each benchmark to use 15 GB of memory and imposed a time limit of 10 minutes of wall-clock time.

\pacmc is written in JAVA and uses \storm~\cite{DBLP:journas/sttt/HenselJKQV21} and \matlab to get the value of the analyzed property and the solution of the LP problem, respectively.
We also used \storm v1.7.0 and \prism~\cite{DBLP:conf/cav/KwiatkowskaNP11} v4.7 to compute the actual rational functions for the benchmarks, to check how well our PAC approximation works in practice.
We were unable to compare with the fraction-free approach proposed in~\cite{DBLP:journals/iandc/BaierHHJKK20} since it is implemented as an extension of \storm v1.2.1 that fails to build on our system.
To avoid to call repeatedly \storm for each sample as an external process, we wrote a C wrapper for \storm that parses the input model and formula and sets the model constants only once, and then repeatedly instantiates the obtained parametric model with the samples and computes the corresponding values of the property, similarly to the batch mode used in~\cite{DBLP:conf/cav/BadingsJJSV22}.
We also implemented a multi-threaded evaluation of the sampled points, by calling multiple instances of the wrapper in parallel on a partition of the samples.

\subsection{Overall Evaluation}
\label{ssec:experimentsOverallEvaluation}

\begin{table}[t]
    \setlength{\tabcolsep}{3pt}
    \centering
    \caption{Overview of the outcomes of the experiments}
    \label{tab:experimentsDtmcsProbabilityOverallParallel}
    \begin{tabular}{cc|c|c|ccccc}
        & \multirow{2}{*}{Outcome} & \multirow{2}{*}{\prism} & \multirow{2}{*}{\storm} & \pacmc[1] & \pacmc[2] & \pacmc[3] & \pacmc[4] & \pacmc[5] \\
        & & & & \multicolumn{5}{c}{\pacmc[d] parallelism: 1 thread/8 threads} \\
        \hline
        \multirow{3}{*}{\rotatebox[origin=c]{90}{$\lP{=?}[\lpf]$}} &
        Success & 522 & 576 & 594/629 & 585/621 & 576/621 & 576/621 & 576/603 \\
        & Memoryout & 18 & 63 & 0/306 & 0/306 & 0/306 & 0/306 & 0/306 \\
        & Timeout & 396 & 297 & 342/1 & 351/9 & 360/9 & 360/9 & 360/27 \\
        \hline
        \multirow{3}{*}{\rotatebox[origin=c]{90}{$\lER{=?}[\lpf]$}} &
        Success & 153 & 224 & 302/302 & 302/302 & 302/302 & 302/302 & 302/302 \\
        & Memoryout & 0 & 0 & 0/282 & 0/282 & 0/282 & 0/282 & 0/282 \\
        & Timeout & 467 & 396 & 318/36 & 318/36 & 318/36 & 318/36 & 318/36 
    \end{tabular}
\end{table}

In Table~\ref{tab:experimentsDtmcsProbabilityOverallParallel} we show the outcome of the different tools on the 936 probabilistic (marked with $\lP{=?}[\lpf]$) and 620 reward (marked with $\lER{=?}[\lpf]$) benchmarks, namely whether they successfully produced a rational function or whether they failed by timeout or by running out of memory.
Besides the results for \prism and \storm computing the actual rational function, we report two values for each outcome of \pacmc[d], where the superscript $d$ indicates the degree of the polynomial used as template:
in e.g. the pair 594/629, the first value 594 is relative to the single-threaded \pacmc[1], while the value 629 is for the 8-threaded \pacmc[1], i.e., \pacmc with 8 instances of the \storm wrapper running in parallel.
As parameters for \pacmc, we set $\errorRate = \significanceLevel = 0.05$; 
for the benchmarks with two parameters, this results in sampling between 280 and 1000 points, for $d = 1$ to $d = 5$, respectively.
To make the comparison between the different templates fairer, we set the same random seed for each run of \pacmc; 
this ensures that all samples used by e.g. \pacmc[2] are also used by \pacmc[5].
As we can see from Table~\ref{tab:experimentsDtmcsProbabilityOverallParallel}, \pacmc is able to compute polynomials with different degrees for more benchmarks than \storm and \prism.
By inspecting the single experiments, for the probabilistic properties we have that $\prism \subseteq \storm \subseteq \pacmc[d]_{n} \subseteq \pacmc[d']_{n}$ for each $d' < d$ degrees and $n$ threads, as sets of successfully solved cases; 
we also have that $\pacmc[d]_{1} \subseteq \pacmc[d]_{8}$ for each $d$.
For the reward properties we have that $\pacmc[d]_{n} = \pacmc[d']_{n'}$ for each combination of $d, d' \in \setnocond{1, \cdots, 5}$ and $n, n' \in \setnocond{1, 8}$ and that $\storm, \prism \subseteq \pacmc[d]_{n}$; 
however \storm and \prism are incomparable, with cases solved by \storm but not by \prism, and vice-versa.
In the next section we will evaluate how the margin $\margin$ changes depending on the degree $d$ and the statistical parameters $\errorRate$ and $\significanceLevel$ through the induced~number~of~samples.

\subsection{Relation of the Polynomial Degree $d$ and the Number of Samples with the Margin $\margin$ and the Distance $\norm[2]{f_{\lsf} - \approxfun{f}_{\lsf}}$}
\label{ssec:experimentsDegreeSamplesVsMarginDistance}

\begin{figure}[t]
	\centering
 	\resizebox{\linewidth}{!}{
    	\begin{tikzpicture}
    	    \node (scatter) at (0,0) {\resizebox{!}{40mm}{\includegraphics{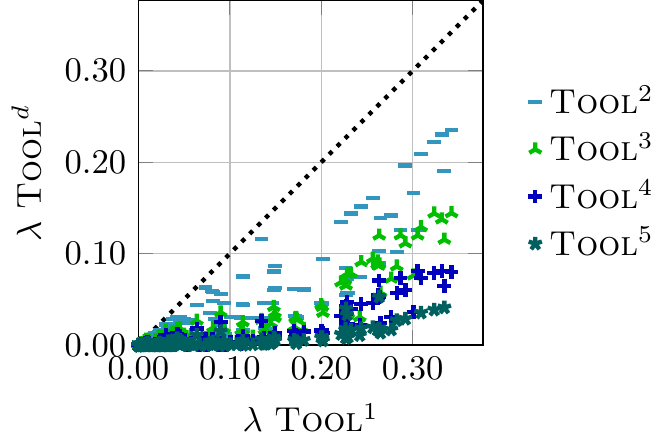}}};
    		\node[anchor=west] (box) at ($(scatter.east) + (0.75,0)$) {\resizebox{!}{40mm}{\includegraphics{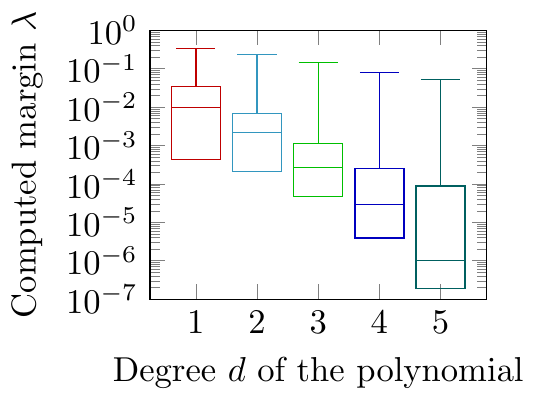}}};
    	\end{tikzpicture}
	}
	\caption{Scatter plot for the margin~$\margin$ for different \pacmc[d] and box plots for the margin~$\margin$}
	\label{fig:experimentsDtmcsProbabilityLambdaAndBoxes}
\end{figure}

In Fig.~\ref{fig:experimentsDtmcsProbabilityLambdaAndBoxes} we present plots for \pacmc using polynomial templates with different degrees and how the computed $\margin$ changes.
As we can see from the plots, by using a higher degree we get a lower value for the margin $\margin$, as one would expect given that polynomials with higher degree can approximate better the shape of the actual rational function: 
from the box plots on the right side of the figure, we can see that using higher degree polynomials allows us to get values for $\margin$ that are much closer to $0$.
Note that in these box plots we removed the lower whiskers since they are $0$ for all degrees, and we use a logarithmic y-axis.
The scatter plot shown on the left side of Fig.~\ref{fig:experimentsDtmcsProbabilityLambdaAndBoxes}, where we compare the values of $\margin$ produced by \pacmc[1] with those by \pacmc[d], for $d=2,3,4,5$, confirms that the higher the degree is, the closer to $0$ the corresponding mark is, since the points for the same benchmark share the same x-axis value.

\begin{figure}[t]
    \centering
    \resizebox{\linewidth}{!}{
        \includegraphics{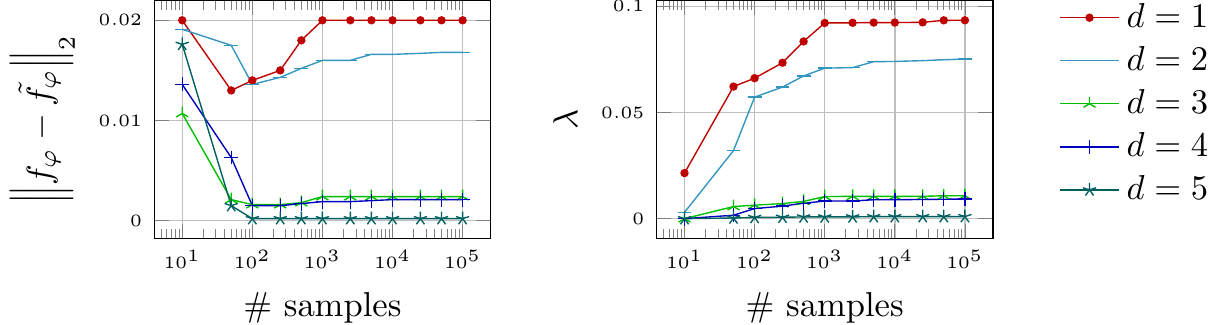}
    }
    \caption{Value of $\norm[2]{f_{\lsf}-\approxfun{f}_{\lsf}}$ and of $\margin$ vs.\@ degree of polynomials and number of samples}
    \label{fig:dtmcProbabilityNormSamples9}
\end{figure}

In Fig.~\ref{fig:dtmcProbabilityNormSamples9} we show the value of $\norm[2]{f_{\lsf} - \approxfun{f}_{\lsf}}$, that is, how close the polynomial $\approxfun{f}_{\lsf}$ is to the actual rational  function $f_{\lsf}$, for different degrees of the polynomial and the number of samples, as well as the corresponding values of the computed $\margin$. 
The plots are relative to one benchmark such that the corresponding rational function (a polynomial having degree 96) computed by \storm can be managed by \matlab without incurring in obvious numerical errors, while having the margin $\margin$ computed by \pacmc[2] reasonably large ($\margin \approx 0.063$).

From the plots we can see that we need at least 100 samples to get a rather stable value for $\norm[2]{f_{\lsf} - \approxfun{f}_{\lsf}}$, so that the value of $\norm[2]{f_{\lsf} - \approxfun{f}_{\lsf}}$ is smaller for higher degrees, which reflects the more accurate polynomial approximation to the original function, in line with the plots in Fig.~\ref{fig:experimentsDtmcsProbabilityLambdaAndBoxes}. 
However, for the same degree, as the number of samples increases, the value of $\norm[2]{f_{\lsf} - \approxfun{f}_{\lsf}}$ does not always decrease. 
This happens because with few points, the polynomial can fit them well, as indicated by the low value of $\margin$; 
however, such few points are likely to be not enough to represent accurately the shape of $f_{\lsf}$.
By increasing the number of samples, the shape of $f_{\lsf}$ can be known better, in particular where it changes more;
this makes it more difficult for the polynomials to approximate $f_{\lsf}$, 
as indicated by the larger $\margin$;
on the other hand, they get closer to $f_{\lsf}$, so $\norm[2]{f_{\lsf} - \approxfun{f}_{\lsf}}$ stabilizes.

\subsection{Relation of the Statistical Parameters $\errorRate$ and $\significanceLevel$ with the Distances $\norm[2]{f_{\lsf} - \beta}$ and $\upperboundFXBeta(\approxfun{f}_{\lsf}, X, \beta)$}
\label{ssec:experimentsStatisticalparametersVsDistancesbeta}

\begin{figure}[t]
    \centering
    \resizebox{\linewidth}{!}{
    \includegraphics{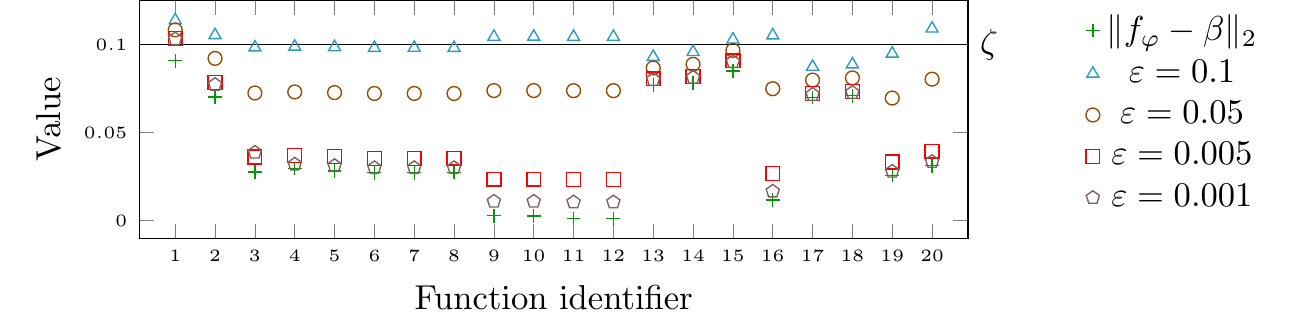}
    }
    \caption{Comparison of $\norm[2]{f_{\lsf} - \beta}$ with $\upperboundFXBeta(\approxfun{f}_{\lsf}, X, \beta)$ for $\significanceLevel = 0.05$ and different $\errorRate$}
	\label{fig:Lemma3ComparewithPolynomials}
\end{figure}

We now consider the behavior of $f_{\lsf}$ and whether it remains close to some number $\beta$ within $\safetyLevel$, that is, we want to check whether $\norm[2]{f_{\lsf} - \beta} < \safetyLevel$ holds.
Here we set the safety level $\safetyLevel$ to be $0.1$ and consider different $\beta$’s values for different functions $f_{\lsf}$. 
We consider 20 rational functions computed by \storm that \matlab can work without incurring in obvious numerical errors, such as those outside the probability interval $\interval{0}{1}$.
For each of the function, we computed the corresponding value of $\beta$ by sampling 20 points for the parameters and taking the average value, rounded to the first decimal, of the function on them.
We rely on Lemma~\ref{lem:nearMargin} to perform the analysis;
the results are shown in Fig.~\ref{fig:Lemma3ComparewithPolynomials}.

In the figure, we plot the actual value of $\norm[2]{f_{\lsf} - \beta}$, the boundary $\safetyLevel$, and the value of $\upperboundFXBeta(\approxfun{f}_{\lsf}, X, \beta)$ computed with respect to $\significanceLevel = 0.05$ and different choices of $\errorRate$ for the 20 functions.
As we can see, the smaller $\errorRate$, the higher the number of cases on which Lemma~\ref{lem:nearMargin} ensures $\norm[2]{f_{\lsf} - \beta} < \safetyLevel$;
this is expected, since a smaller $\errorRate$ increases the number of samples, so the approximating polynomial $\approxfun{f}_{\lsf}$ gets closer to the real shape of $f_{\lsf}$.
Moreover, when $\norm[2]{f_{\lsf} - \beta}$ is already close to $\safetyLevel$, there is little space for $\approxfun{f}_{\lsf}$ to differ from $f_{\lsf}$, as happens for the e.g.\@ the function 1.
Thus it is more difficult for us to be able to rely on Lemma~\ref{lem:nearMargin} to check whether $\norm[2]{f_{\lsf} - \beta} < \safetyLevel$ holds, even if this actually the case.

\subsection{Comparison with the Taylor Expansion}
\label{subsec:comparewithTaylor}

\begin{figure}[t]
    \centering
    \resizebox{\linewidth}{!}{
        \includegraphics{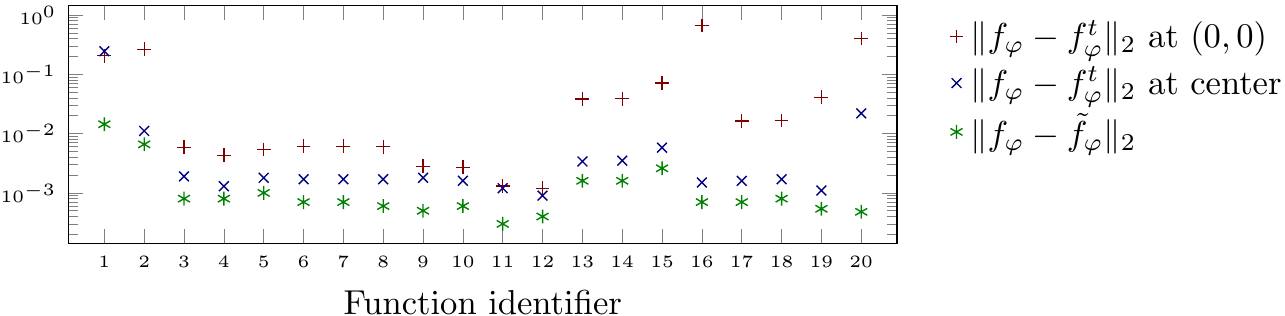}
    }
    \caption{Distance from $f_{\lsf}$ of the Taylor expansion vs.\@ the approximating polynomial}
    \label{fig:distanceTaylor}
\end{figure}
We compare the accuracy of PAC approximation against that of the Taylor expansion on the same cases used for Fig.~\ref{fig:Lemma3ComparewithPolynomials}; 
the comparison is shown in Fig.~\ref{fig:distanceTaylor}. 
For the comparison with $f_{\lsf}$, we consider the degree 2 for both the Taylor expansion $f^{t}_{\lsf}$ and the approximating polynomial $\approxfun{f}_{\lsf}$ computed with $\errorRate = \significanceLevel = 0.05$.
For the Taylor expansion $f^{t}_{\lsf}$, we considered two versions: 
the expansion at the origin, i.e., $(0,0)$ for two parameters (marked as ``$\norm[2]{f_{\lsf} - f^{t}_{\lsf}}$ at $(0,0)$'' in Fig.~\ref{fig:distanceTaylor}), that is commonly used since it is cheaper to compute than the expansions at other points; 
and 
the expansion at the barycenter of the space of the parameters (marked as ``$\norm[2]{f_{\lsf} - f^{t}_{\lsf}}$ at center'' in Fig.~\ref{fig:distanceTaylor}).

As we can see from the plot, that uses a logarithmic scale on the y-axis, the distance $\norm[2]{f_{\lsf} - \approxfun{f}_{\lsf}}$ is between one and three orders of magnitude smaller than $\norm[2]{f_{\lsf} - f^{t}_{\lsf}}$ at the origin.
If we consider $\norm[2]{f_{\lsf} - f^{t}_{\lsf}}$ at the barycenter, we get values much closer to $\norm[2]{f_{\lsf} - \approxfun{f}_{\lsf}}$, but still larger up to one order of magnitude.
One of the reasons for this is that the Taylor expansion reflects local properties of $f_{\lsf}$ at the expansion point, while the PAC approximation provides a global approximation of $f_{\lsf}$, thus reducing the overall distance.
Compared with the Taylor expansion, the PAC approximation has also other advantages:
the PAC approximation can handle both white-box and black-box problems, i.e., we do not need to get the analytical form of $f_{\lsf}$; 
this means that we can treat it as a black box and get a good approximation of it while the Taylor expansion can only be applied after computing the actual function $f_{\lsf}$. 
Moreover, the PAC approximation is able to generate polynomials with any given error rate and provide probabilistic guarantee, while Taylor expansion cannot.

\subsection{Extension to Reward Models}

\begin{figure}[t]
    \centering
    \resizebox{\linewidth}{!}{
        \includegraphics{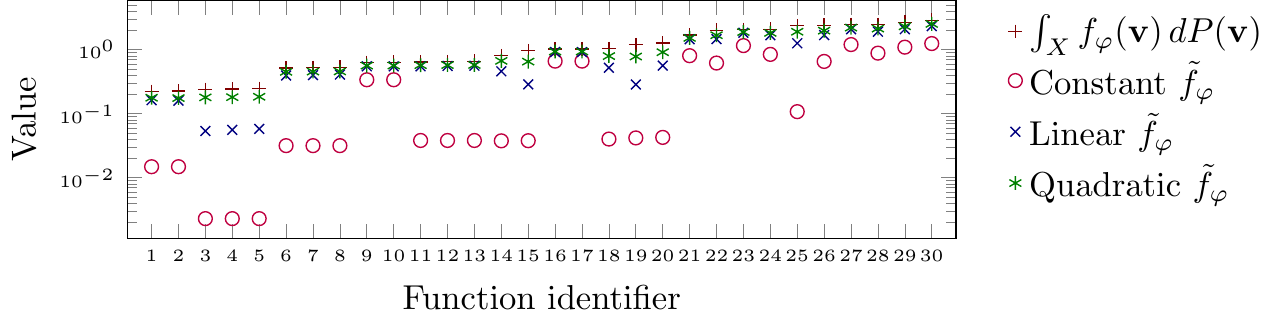}
    }
    \caption{Lower bound for Eq.~\eqref{eq:rewardmodel} by PAC approximation with different degrees}
    \label{fig:RewardModel}
\end{figure}

In Fig.~\ref{fig:RewardModel} we show how Eq.~\eqref{eq:lowerboundofrewardmodel} applies to $\int_{X} f_{\lsf}(\parameters)\,dP(\parameters)$ for a selection of 30 reward properties $f_{\lsf}$ computed by \storm; 
as usual, we compute $\approxfun{f}_{\lsf}$ with $\errorRate = \significanceLevel = 0.05$.
In the figure, we report the actual value of $\int_{X} f_{\lsf}(\parameters)\,dP(\parameters)$ as well as that of the expression in Eq.~\eqref{eq:lowerboundofrewardmodel} computed for the polynomial PAC approximations $\approxfun{f}_{\lsf}$ at different degrees.
As we can see from Fig.~\ref{fig:RewardModel}, the higher the degree of $\approxfun{f}_{\lsf}$, the more accurate the estimation of the $\int_{X} f_{\lsf}(\parameters)\,dP(\parameters)$'s lower bound is.
In particular, the quadratic $\approxfun{f}_{\lsf}$ provides a very close lower bound for $\int_{X} f_{\lsf}(\parameters)\,dP(\parameters)$;
this is remarkable, since evaluating $\max(\approxfun{f}(\parameters) - \margin)$ in Eq.~\eqref{eq:lowerboundofrewardmodel} is often an NP-hard non-convex optimization problem~\cite{DBLP:journals/siamcomp/Sahni74,pardalos1991algorithms} and, for cubic or higher polynomials, it requires specialized theories and tools to solve~\cite{lasserre2009moments,yang2022computing,DBLP:journals/mp/Lasserre08}.

\section{Conclusion}
\label{sec:conclusion}

In this paper, we presented a PAC-based approximation framework for studying several properties of parametric discrete time Markov chains.
Within the framework, we can analyze the safety regions of the domain of the parameters, check whether the actual probability fluctuates around a reference value within a certain bound, and get a polynomial approximating the actual probability rational function with given $(\errorRate, \significanceLevel)$-PAC guarantee.
An extended experimental evaluation confirmed the efficacy of our framework in analyzing parametric models.

As future work, we plan to investigate the applicability of the scenario approach to other Markov models and properties, such as continuous time Markov chains and Markov decision processes with and without rewards, where parameters can also control the rewards structures.
Moreover, we plan to explore the combination of the scenario approach with statistical model checking and black-box verification and model learning.

\subsubsection{Acknowledgements.}
We thank the anonymous reviewers for their useful remarks that
helped us improve the quality of the paper.
Work supported in part by 
the CAS Project for Young Scientists in Basic Research under grant No.\@ YSBR-040,
NSFC under grant No.\@ 61836005,
the CAS Pioneer Hundred Talents Program,
the ISCAS New Cultivation Project ISCAS-PYFX-202201,
and
the ERC Consolidator Grant 864075 (\emph{CAESAR}).
\newline\protect\includegraphics[height=8pt]{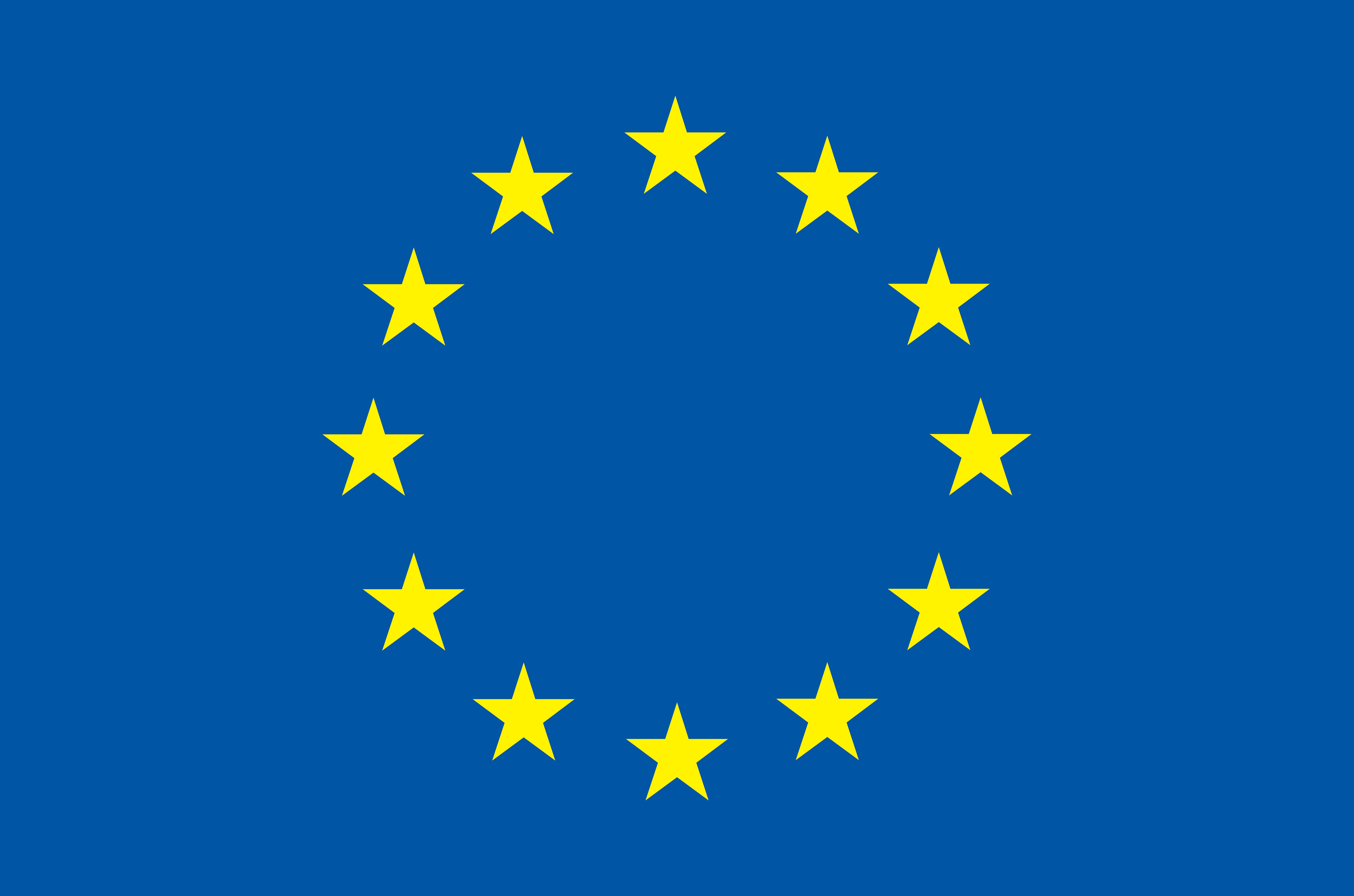} This project is part of the European Union’s Horizon 2020 research and innovation programme under the Marie Sk\l{}odowska-Curie grant no.\@ 101008233.

\paragraph{Data Availability Statement.}
An environment with the tools and data used for the experimental evaluation presented in this work is available in the following Zenodo repository:
\url{https://doi.org/10.5281/zenodo.8181117}.

\bibliographystyle{splncs04}
\bibliography{biblio}

\newpage
\appendix
\section{Proofs of the Lemmas}
\label{app:proofs}

\lemLinearApproximation*
\begin{proof}
On the one hand, if the condition $\approxfun{f}(\parameters) + \margin < \safetyLevel$ holds for each $\parameters \in X$, then we have 
\begin{align*}
    P(f(\parameters) < \safetyLevel) 
    & \geq P(f(\parameters) < \approxfun{f}(\parameters) + \margin)\\
    & {} \geq P(f(\parameters) - \approxfun{f}(\parameters) < \margin) \\
    & {} \geq P(\abs{f(\parameters) - \approxfun{f}(\parameters)} < \margin).
\end{align*}
By the definition of PAC approximation, it follows that
\[
    P(f(\parameters) < \safetyLevel) \geq 1 - \errorRate
\]
so the parameters space $X$ is safe with the confidence of $1 - \significanceLevel$. 

On the other hand, we first assume that for each $\parameters \in X$, the condition $f(\parameters) < \safetyLevel$ holds. 
Since $\approxfun{f}$ is a PAC approximation of $f$ with $(\errorRate, \significanceLevel)$-guarantee, this implies that
\[
    P(\abs{\approxfun{f}(\parameters) - f(\parameters)} \leq \margin) \geq 1 - \errorRate
\]
according to Def.~\ref{def:PACmodel}. 
Moreover, we have 
\[
    P(f(\parameters) - \margin \leq \approxfun{f}(\parameters) \leq f(\parameters) + \margin) \geq 1 - \errorRate.
\]
Therefore, 
\[
    P(\approxfun{f}(\parameters) \leq f(\parameters) + \margin) \geq 1 - \errorRate
\] 
holds. 
Since $f(\parameters) < \safetyLevel$, this implies that 
\[
    P(\approxfun{f}(\parameters) \leq \safetyLevel + \margin) \geq 1 - \errorRate, 
\]
which is equivalent to the following inequality:
\[  
    P(\approxfun{f}(\parameters) > \safetyLevel + \margin) < \errorRate.
\] 
This contradicts the assumption ``$P(\approxfun{f}(\parameters) -\margin > \safetyLevel) > \errorRate$'' in the statement of the lemma, thus the condition we assumed ``$\forall \parameters \in X$, the condition $f(\parameters) < \safetyLevel$ holds'' cannot be true. 
From this we derive that there exists a point $\parameters \in X$ such that $f(\parameters) > \safetyLevel$ with confidence $1 - \significanceLevel$, i.e., the domain of parameters $X$ is unsafe, as desired.
\end{proof}

\lemNearMargin*
\begin{proof}
In the following sequence of (in)equalities, we motivate between them how to obtain the next term in the sequence.
\begin{align*}
        & \norm[p]{f - \beta} \\
        \intertext{By definition of the $L_{p}$ norm}
        = & \sqrt[p]{\int_{X} \abs{f(\parameters) - \beta}^{p}\,d\parameters}\\
        \intertext{By splitting the integral region into two parts}
        = & \sqrt[p]{\int_{X_{1}} \abs{f(\parameters) - \beta}^{p}\,d\parameters + \int_{X_{2}} \abs{f(\parameters) - \beta}^{p}\,d\parameters}  \\
        \intertext{Where $X = X_{1} \uplus X_{2}$ and $X_{1}$ is such that $P(X_{1}) \geq 1 - \errorRate$ and for each $\parameters \in X_{1}$, we have $\abs{f(\parameters) - \approxfun{f}(\parameters)} \leq \lambda$. 
        Then by known triangular inequality of the $L_{p}$ norm}
        \leq & \sqrt[p]{\left( \sqrt[p]{\int_{X_{1}} \abs{f(\parameters) - \approxfun{f}(\parameters)}^{p}\,d\parameters} + \sqrt[p]{\int_{X_{1}} \abs{\approxfun{f}(\parameters) - \beta}^{p}\,d\parameters} \right)^{p} + \int_{X_{2}} \abs{f(\parameters) - \beta}^{p}\,d\parameters}\\
        \intertext{From the condition $P(\abs{f(\parameters) - \approxfun{f}(\parameters)} \leq \margin) \geq 1 - \errorRate$}
        \leq & \sqrt[p]{\left(\sqrt[p]{(1 - \errorRate) \size{X} \margin^{p}} + \sqrt[p]{\int_{X} \abs{\approxfun{f}(\parameters) - \beta}^{p}\,d\parameters} \right)^{p} + \int_{X_{2}} \abs{f(\parameters) - \beta}^{p}\,d\parameters}\\
        \leq & \sqrt[p]{\left(\margin \sqrt[p]{(1 - \errorRate) \cdot \size{X} } + \norm[p]{\approxfun{f} - \beta} \right)^{p} + \errorRate \cdot \size{X} \cdot \max(\abs{M - \beta}^{p}, \beta^{p})}. 
\end{align*}
Since by the lemma assumption we have that 
\[
    \sqrt[p]{\left(\margin \sqrt[p]{(1 - \errorRate) \cdot \size{X} } + \norm[p]{\approxfun{f} - \beta} \right)^{p} + \errorRate \cdot \size{X} \cdot \max(\abs{M - \beta}^{p}, \beta^{p})} < \safetyLevel
\] 
holds, it follows that the property $\norm[p]{f - \beta} < \safetyLevel$ is satisfied as well, with confidence $1-\significanceLevel$.
\end{proof}

\lemRewardModel*
\begin{proof}
By splitting the integral region $X$ into two parts $X = X_{1} \uplus X_{2}$ where $X_{1}$ is such that $P(X_{1}) \geq 1 - \errorRate$ and for each $\parameters \in X_{1}$ we have $\abs{f(\parameters) - \approxfun{f}(\parameters)} \leq \margin$, as in the proof of Lemma~\ref{lem:nearMargin}, we have
\[
    \int_{X} f(\parameters)\,dP(\parameters) = \int_{X_{1}} f(\parameters)\,dP(\parameters) + \int_{X_{2}} f(\parameters)\,dP(\parameters).
\]
Since we have $f(\parameters) \geq 0$ for each $\parameters \in X$ by definition of $f$, it follows that $\int_{X_{2}} f(\parameters)\,dP(\parameters) \geq 0$. 
This implies that
\begin{align*}
    \int_{X} f(\parameters)\,dP(\parameters) & \geq \int_{X_{1}} f(\parameters)\,dP(\parameters) \\
    & \geq \int_{X_{1}} (\approxfun{f}(\parameters) - \margin)\,dP(\parameters) \\
    & = \int_{X} (\approxfun{f(\parameters)} - \margin)\,dP(\parameters) - \int_{X_{2}} (\approxfun{f}(\parameters) - \margin)\,dP(\parameters) \\
    & \geq \int_{X} (\approxfun{f}(\parameters) - \margin)\,dP(\parameters) - \errorRate \cdot \size{X} \cdot \max_{\parameters \in X_{2}}(\approxfun{f}(\parameters) - \margin) \\ 
    & \geq \int_{X} (\approxfun{f}(\parameters) - \margin)\,dP(\parameters) - \errorRate \cdot \size{X} \cdot \max_{\parameters \in X}(\approxfun{f}(\parameters) - \margin).
\end{align*}
This means that if the approximation polynomial $\approxfun{f}$ satisfies 
\[
    \int_{X} (\approxfun{f}(\parameters) - \margin)\,dP(\parameters) - \errorRate \cdot \size{X} \cdot \max_{\parameters \in X}(\approxfun{f}(\parameters) - \margin) \geq \rewardLevel,
\]
then 
\[
    \int_{X} f(\parameters)\,dP(\parameters) \geq \int_{X} (\approxfun{f}(\parameters) - \margin)\,dP(\parameters) - \errorRate \cdot \size{X} \cdot \max_{\parameters \in X}(\approxfun{f}(\parameters) - \margin) \geq \rewardLevel
\]
holds with confidence $1 - \significanceLevel$, as required.
\end{proof}

\end{document}